\documentclass{sig-alternate}

\usepackage[utf8]{inputenc}
\usepackage{ifpdf}
\usepackage{amsmath}
\usepackage{amssymb}
\usepackage{latexsym}
\usepackage{mathtools}
\usepackage{color}

 \newtheorem{thm}{Theorem}
 \newtheorem{lem}[thm]{Lemma}
 \newtheorem{prop}[thm]{Proposition}
 \newtheorem{cor}[thm]{Corollary}
 \newtheorem{defn}{Definition}

\def\KK{{\mathbb{K}}}

\def\PP{{\mathbb{P}}}
\def\ZZ{{\mathbb{Z}}}
\def\QQ{{\mathbb{Q}}}
\def\NN{{\mathbb{N}}}
\def\UU{{\mathbb{U}}}
\def\Res{{\mathrm{Res}}}
\def\Jac{{\mathrm{Jac}}}

\def\Disc{{\mathrm{Disc}}}
\def\Spec{{\mathrm{Spec}}}

\def\Proj{{\mathrm{Proj}}}

\def\P{{\mathfrak{P}}}

\def\Rc{{\mathcal{R}}}

\def\Dc{{\mathcal{D}}}
\def\Jc{{\mathcal{J}}}

\def\T{{\mathfrak{T}}}

\def\ux{{\underline{x}}}
\def\mm{{\mathfrak{m}}}

\title{Discriminants of complete intersection space curves}

\numberofauthors{2}
\author{
\alignauthor Laurent Bus\'e\\
       \affaddr{Universit\'e C\^ote d'Azur, Inria}\\
       \affaddr{2004 route des Lucioles,}\\
       \affaddr{06 902 Sophia Antipolis, France}\\
       \email{laurent.buse@inria.fr}
\alignauthor Ibrahim Nonkan\'e\\
       \affaddr{Université Ouaga 2, IUFIC,}\\
	   \affaddr{12 BP 417 Ouagadougou 12, Burkina Faso}\\
       \email{inonkane@univ-ouaga2.bf}
}

\begin{document}

\maketitle

\begin{abstract}
In this paper, we develop a new approach to the discriminant of a complete intersection curve in the 3-dimensional projective space. By relying on the resultant theory, we first prove a new formula that allows us to define this discriminant without ambiguity and over any commutative ring, in particular in any characteristic. This formula also provides a new method for evaluating and computing this discriminant efficiently, without the need to introduce new variables as with the well-known Cayley trick. Then, we obtain new properties and computational rules such as the covariance and the invariance formulas. Finally, we show that our definition of the discriminant satisfies to the expected geometric property and hence yields an effective smoothness criterion for complete intersection space curves. Actually, we show that in the generic setting, it is the defining equation of the discriminant scheme if the ground ring is assumed to be a unique factorization domain.
\end{abstract}

\section{Introduction}

Discriminants are central mathematical objects that have applications in many fields. Let $\KK$ be a field and suppose given integers $1\leq c \leq n$ and $1\leq d_1,\ldots,d_c$. Let $S$ be the set of all $c$-uples of homogeneous polynomials $f_1,\ldots,f_c$ in the polynomial ring $\KK[x_1,\ldots,x_n]$ of degree $d_1,\ldots,d_c$ respectively. Consider the subset $D$ of $S$ corresponding to those $c$-uples of homogeneous polynomials that define an algebraic subvariety in $\PP^{n-1}_\KK$ which is not smooth  and of codimension $c$. It is well-known that $D$ is an irreducible hypersurface provided $d_i\geq 2$ for some $i$, or provided $c=n$ (in which case $D$ is nothing but the resultant variety) \cite{GKZ}. The discriminant polynomial is then usually defined as an equation of $D$. It is a homogeneous polynomial in the coefficients of each polynomial $f_i$ whose vanishing provides a smoothness criterion \cite{GKZ,Benoist}.  
This geometric approach to discriminants yields a beautiful theory with many remarkable results (e.g.~\cite{GKZ}). However, whereas there are strong interests in computing with discriminants (e.g.~\cite{Dem,Sase,Nie12,GM15,BK16}), including in the field of number theory, this approach is not tailored to develop the required formalism. For instance, having the discriminant defined up to a nonzero multiplicative constant is an important drawback, especially when computing over fields of positive characteristic. Another point is about the computation of discriminants: it is usually done by means of the famous Cayley trick that requires to introduce new variables, which has a bad effect on the computational cost. 

\smallskip

In some cases, there exist an alternative to the above geometric definition of discriminants. In the case $c=n$, which corresponds to resultants, there is a huge literature where the computational aspects are treated extensively. In particular, a vast formalism is available and many formulas allow to compute resultants, as for instance the well-known Macaulay formula (e.g.~\cite{J91,J97,CLO98}). When $c<n$ the theory becomes much more delicate. Nevertheless, for both cases $c=1$ (hypersurfaces) and $c=n-1$ (finitely many points) discriminants can be defined rigorously and their formalism has been developed. The case $c=1$ goes back to Demazure \cite{Bou, Dem} and the case $c=n-1$ has been initiated by Krull \cite{Krull1,Krull2}. In both cases, the discriminant is defined by means of resultants, via a universal formula. This allows to develop the formalism, to obtain useful computational rules and also to compute it efficiently by taking advantage of the Macaulay formula for resultants; see \cite{BJ14} for more details.

The goal of this paper is to provide a similar treatment in the case $(c,n)=(2,4)$. Our approach relies on the characterization of this discriminant by means of a universal formula where resultants and discriminants of finitely many points appear. As far as we know, this formula is new and provide the first (efficient) method to compute the discriminant of a complete intersection curve over any ring. In particular, we provide a closed formula that allows to compute it as a ratio of determinants. We emphasize that the computations are done in dimension at most 3, that is to say that there is no need to introduce new variables as with the Cayley trick. We mention that the problem of studying and computing discriminants goes back to the remarkable paper \cite{Sylvester} of Sylvester in 1864. The case $(c,n)=(2,4)$ was the last remaining case to complete the picture in $\PP^3$. 

\smallskip

Before going into further details, we provide an example to illustrate the contribution of this paper. 
The Clebsch cubic projective surface  is defined by the homogeneous polynomial
\vspace{-.2em}
$$f_1:=\frac{1}{3}\left(\sum_{i=1}^4 x_i^3 - (\sum_{i=1}^4 x_i)^3 \right)\in \ZZ[x_1,x_2,x_3,x_4].$$
By \cite[Definiton 4.6]{BJ14} and the Macaulay formula, we get
$$3^5 \cdot \Disc(f_1)=\Res(\partial_1f_1,\ldots,\partial_4f_1)=-3^{5}\times 5.$$
Thus, $\Disc(f_1)=-5$ and we recover that the Clebsch surface is smooth except in characteristic 5. 
Now, consider the family of quadratic forms
	$$f_2 := ax_1^2+x_1x_2+x_2^2+x_3^2+x_4^2 \in \ZZ[a][x_1,x_2,x_3,x_4].$$
The formula \eqref{eqn:def} we will prove in this paper allows to compute the discriminant of the intersection curve between the Clebsch surface and these quadratic forms;  we get 
\begin{multline*}
		\Disc(f_1,f_2)=2a \left( 110592{a}^{7}+442944{a}^{6}+1163408{a}^{5}+ \right. \\
\left. 	1303260 {a}^{4}+416575{a}^{3}+238468{a}^{2}-33924a-4448 \right) \\ 
\cdot \left( 
	3456{a}^{5}+8208{a}^{4}+10656{a}^{3}+14069{a}^{2}+11134a+
	3176 \right) ^{2}.
\end{multline*}
In characteristic 5, the Clebsch surface $f_1=0$ is singular at the point $P=(1:1:1:1)$. So, if the surface defined by the equation $f_2=0$ goes through $P$ then their intersection curve will be singular at $P$. In general, this is not the case. Indeed, we have that
\begin{multline*}
	\Disc(f_1,f_2)=2a \left( 2{a}^{7}+4{a}^{6}+3{a}^{5}+3{a}^{2}+a+2 \right)\\
	 \cdot \left( {a}^{5}+3{a}^{4}+{a}^{3}+4{a}^{2}+4a+1 \right) ^{2}  \mod 5.
\end{multline*}
Now, if $a$ is specialized to $5b-4$ then we force the surfaces defined by $f_2=0$ to go through $P$. Applying this specialization  the above formula, we obtain
\begin{multline*}
\Disc(f_1,{f_2}_{|a=5b-4})=2\cdot \left( 5b-4 \right)\cdot 5  \\
	 \cdot \left( 31250{b}^{7}-\ldots \right) 
	 \left( 3125{b}^{5}-\ldots \right)^{2} \mod 5
\end{multline*}
so that this discriminant now vanishes modulo 5 as expected.

\smallskip

The paper is organized as follows. In Section \ref{sec:def} we prove a new formula, based on resultants, that is used to provide a new definition of the discriminant of a complete intersection space curve. Then, in Section \ref{sec:formulas} we give some properties and computational rules of this discriminant by relying on the existing formalism of resultants. Finally, in Section \ref{sec:geom} we show that our definition is correct in the sense that it satisfies to the expected geometric property, in particular it yields a universal and effective smoothness criterion which is valid in arbitrary characteristic. 

In the sequel, we will rely heavily on the theory of resultants and its formalism, including the Macaulay formula. We refer the reader to \cite{J91} and \cite[Chapter 3]{CLO98}. We will also assume some familiarity with the definition of discriminants in the case $c=n-1$ for which we refer the reader to \cite[\S 3.1]{BJ14}. Resultants and discriminants will be denoted by $\Res(-)$ and $\Disc(-)$ respectively.

\section{Definition and formula}\label{sec:def}

Suppose given two positive integers $d_1,d_2$ and consider the generic homogeneous polynomials in the four variables $\ux=(x_1,x_2,x_3,x_4)$ 
$$f_1:=\sum_{|\alpha|=d_1} U_{1,\alpha}x^\alpha, \ \ f_2:=\sum_{|\alpha|=d_2} U_{2,\alpha}x^\alpha.$$
We denote by $A=\ZZ[U_{i,\alpha}; i=1,2, |\alpha|=d_i]$ the universal ring of coefficients and we define the polynomial ring $C=A[\ux]$.
The partial derivative of the polynomial $f_i$ with respect to the variable $x_j$ will be denoted by $\partial_jf_i$. Moreover, given four homogeneous polynomials $p_1,$ $p_2,$ $p_3,$ $p_4$ in the variables $\ux$, the determinant of their Jacobian matrix will be denoted by 
$J(p_1,p_2,p_3,p_4):=\det\left(\partial_jp_i \right)_{i,j=1\ldots,4}.$

\begin{thm}\label{thm:gendef} Using the above notation, assume that $d_1+d_2\geq 3$. Let $l,m,n$ be three linear forms 
$$l(\ux)=\sum_{i=1}^4l_ix_i, \ m(\ux)=\sum_{i=1}^4m_ix_i, \ n(\ux)=\sum_{i=1}^4n_ix_i,$$	
and denote by $A'$ the polynomial ring extension of $A$ with the coefficients $l_i$'s, $m_j$'s and $n_k$'s of the linear forms $l,m,n$. Then, 
there exists a unique polynomial in $A$, denoted by $\Disc(f_1,f_2)$ and called the \emph{universal discriminant} of $f_1$ and $f_2$, which is independent of the coefficients of $l,m,n$ and that satisfies to the following equality in $A'$:
\begin{multline*}
	\Res\left(f_1,f_2, J(f_1,f_2,l,m), J(f_1,f_2,l,n)\right) = \Disc(f_1,f_2)  \\
	\hspace{.7cm} \cdot \Res\left(f_1,J(f_1,l,m,n)),f_2,J(f_2,l,m,n)\right)\Disc\left({f_1},{f_2},l\right).
\end{multline*}
By convention, if $d_j=1$ we set 
$$\Res\left(f_1,J(f_1,l,m,n)),f_2,J(f_2,l,m,n)\right)
=J(f_j,l,m,n)^{D_j}$$
where $D_j=(d_1+d_2-d_j)(d_1+d_2-d_j-1)$.
\end{thm}

Given a commutative ring $R$ and two homogeneous polynomials 
$$g_1:=\sum_{|\alpha|=d_1} u_{1,\alpha}x^\alpha, \ \ g_2:=\sum_{|\alpha|=d_2} u_{2,\alpha}x^\alpha$$
in $R[\ux]$ of degree $d_1$, $d_2$ respectively, the map of rings $\rho$ from $A[\ux]$ to $R[\ux]$ which sends $U_{i,\alpha}$ to $u_{i,\alpha}$ and leave each variable $x_i$ invariant, is called the \emph{specialization map} of the universal polynomials $f_1,f_2$ to the polynomials $g_1,g_2$, as  $\rho(f_i)=g_i$.

\begin{defn}\label{def:disc} Suppose given a commutative ring $R$, two positive integers $d_1,d_2$ such that $d_1+d_2\geq 3$ and two homogeneous polynomials $g_1,g_2$ in $R[\ux]$ of degree $d_1,d_2$ respectively. Denoting by $\rho$ the specialization map as above, we define the discriminant of the polynomials $g_1,g_2$ as
	$$\Disc(g_1,g_2)=\Disc(\rho(f_1),\rho(f_2)):=\rho\left(\Disc(f_1,f_2) \right) \in R.$$
\end{defn}

\begin{proof}[of Theorem \ref{thm:gendef}] To prove the claimed formula, one can assume that $A'$ is the universal ring of the coefficients of the polynomials $f_1,f_2,l,m,n$ over the integers.
		
Our \emph{first step} is to show that $\Disc(f_1,f_2,l)$ divides	
$$\Rc:=\Res(f_1,f_2,J(f_1,f_2,l,m), J(f_1,f_2,l,n)).$$
For that purpose, denote by $\Dc$ the ideal of $A[\ux]$ generated by $f_1, f_2$ and all the 3-minors of the Jacobian matrix of the polynomials $f_1,f_2,l$. We also define the ideal $\mm=(\ux)$ and we recall from \cite[Theorem 3.23]{BJ14} that $\Disc(f_1,f_2,l)$ is a generator of the ideal of inertia forms of $\Dc$, i.e.~the ideal 
$$(\Dc:\mm^\infty)\cap A=\{p\in A \textrm{ such that } \exists \nu\in \NN : \mm^\nu\cdot p\subset \Dc \}.$$
Now, from the similar characterization of the resultant by means of inertia forms \cite[Proposition 2.3]{J91}, we deduce that there exists an integer $N$ such that
\begin{multline*}
	\mm^N\cdot \Rc  \subset \left(f_1,f_2,J(f_1,f_2,l,m), J(f_1,f_2,l,n)\right) \subset A[\ux].
\end{multline*}
But $J(f_1,f_2,l,m)$ and $J(f_1,f_2,l,n)$ belong to $\Dc$, so we deduce that 
$\Rc \in (\Dc:\mm^\infty)\cap A.$
It follows that $\Rc$ is an inertia form of $\Dc$ and it is hence divisible by $\Disc(f_1,f_2,l)$.

Our \emph{second step} is to prove that the resultant
$$\Rc_0:=\Res\left(f_1,J(f_1,l,m,n),f_2, J(f_2,l,m,n)\right)$$
 divides $\Rc$. 	
For all $i=1,\ldots,4$, we obviously have that
 $$ \det \left(\begin{array}{ccccc}
\partial_i f_1  & \partial_1 f_1 & \partial_2 f_1 & \partial_3 f_1 & \partial_4 f_1\\
\partial_i f_2  & \partial_1 f_2 & \partial_2 f_2 & \partial_3 f_2 & \partial_4 f_2\\
 l_i & l_1& l_2 & l_3 & l_4 \\
m_i & m_1 & m_2 & m_3 & m_4 \\
n_i & n_1 & n_2 & n_3 & n_4
\end{array}\right) = 0.$$
By developing each of these determinants with respect to their first column, we get the linear system
\begin{multline*}	
\left(\begin{array}{ccc}
l_1 & m_1 & n_1 \\
l_2 & m_2 & n_2 \\
l_3 & m_3 & n_3 \\
l_4 & m_4 & n_4
\end{array}\right)
\left(\begin{array}{c}
	J(f_1,f_2,m,n) \\
	-J(f_1,f_2,l,n) \\
	J(f_1,f_2,l,m)
\end{array}\right)= \\
\left(
\begin{array}{c}
\partial_1f_2 \\
\partial_2f_2 \\
\partial_3f_2 \\
\partial_4f_2
\end{array}
\right)J(f_1,l,m,n)
- \left(
\begin{array}{c}
	\partial_1f_1 \\
	\partial_2f_1 \\
	\partial_3f_1 \\
	\partial_4f_1
\end{array}
\right)J(f_2,l,m,n).
\end{multline*}
The matrix of this linear system is nothing but the transpose of the Jacobian matrix of the polynomials $l,m,n$. Denote by $\Delta$ any of its 3-minor. Then, Cramer's rules show that  
both polynomials $\Delta\cdot J(f_1,f_2,l,m)$ and $\Delta \cdot J(f_1,f_2,l,n)$ belong to the ideal generated by the polynomials $J(f_1,l,m,n)$ and $J(f_2,l,m,n)$. Therefore, the divisibility property of resultants \cite[\S 5.6]{J91} implies that
$\Rc_0$ divides 
\begin{multline*}
\Res(f_1,f_2,\Delta J(f_1,f_2,l,m),\Delta J(f_1,f_2,l,n)) =
\Delta^{r}\cdot \Rc,	
\end{multline*}
where $r=2d_1d_2(d_1+d_2-2)$; observe that $\Delta$ is independent of $\ux$. As it is well-known, $\Delta$ is an irreducible polynomial, being the determinant of a matrix of indeterminates. Therefore, to conclude this second step we have to show that $\Delta$ does not divide  $\Rc_0$. For that purpose, we consider the specialization $\eta$ of the coefficients of $f_1$ and $f_2$ so that 
\begin{equation}\label{eq:eta}
\eta(f_1)=\prod_{i=1}^{d_1} p_i(\ux), \ \ \eta(f_2)=\prod_{i=1}^{d_2} q_i(\ux),	
\end{equation}
where the $p_i$'s and $q_j$'s are generic linear forms; we add their coefficients as new variables to $A'$. Using the multiplicativity property of resultants, a straightforward computation yields the following irreducible factorization formula
\begin{multline}\label{spelinformR0}
\eta(\Rc_0)=\left(\prod_{i=1}^{d_1} J(p_i,l,m,n) \right)^{d_2(d_2-1)} \\ 
\cdot \left(\prod_{i=1}^{d_2} J(q_i,l,m,n) \right)^{d_1(d_1-1)} \cdot \prod_{i,j,r,s} \Res(p_i,p_j,q_r,q_s)^4
\end{multline}
where the last product runs over the integers $i,j=1,\ldots,d_1$, with $i<j$ and $r,s=1,\ldots,d_2$ with $r<s$.
Since $\eta(\Delta)=\Delta$ and $\Delta$ is not a factor in the above formula, we deduce that $\Delta$ does not divide $\Rc_0$.

The \emph{third step} in this proof is to show that the discriminant $\Dc_\infty:=\Disc(f_1,f_2,l)$ and the resultant $\Rc_0$ are coprime polynomials in $A'$. Since $\Dc_\infty$ is irreducible \cite[Theorem 3.23]{BJ14}), we have to show that it does not divide $\Rc_0$. Consider again the specialization $\eta$ given by \eqref{eq:eta} and assume that $\Dc_\infty$ is a factor in $\Rc_0$. Then, since $\Dc_\infty$ is independent on the coefficients of the linear forms $m$ and $n$, $\eta(\Rc_0)$ must contain some factors that depend on the coefficient of $l$ but not on $m$ and $n$. However, the decomposition formula \eqref{spelinformR0} shows that $\eta(\Rc_0)$ contains only irreducible factors that do depend on three linear forms $l,m,n$, or on none of them. Therefore, we deduce that  $\Dc_\infty$ does not divide $\Rc_0$.

To conclude this proof, we observe that the previous results show that $\Dc_\infty \Rc_0$ 
divides $\Rc$. Moreover, straightforward computations shows that $\Dc_\infty \Rc_0$ and $\Rc$ are both homogeneous polynomials with respect to the coefficients of $l$ of the same degree, and the same happens to be true with respect to the coefficients of $m$ and $n$.
\end{proof}

To compute the discriminant it is much more efficient to specialize the formula in Theorem \ref{thm:gendef} by giving to the linear forms $l,m,n$ some specific values, for instance a single variable. Consider
the Jacobian matrix associated to the polynomials $f_1,f_2$ 
$$\Jac(f_1,f_2):=
\left(
\begin{array}{cccc}
	\partial_1 f_1 & \partial_2 f_1 & \partial_3 f_1 & \partial_4 f_1 \\
	\partial_1 f_2 & \partial_2 f_2 & \partial_3 f_2 & \partial_4 f_2	
\end{array}
\right)
$$
and its minors that we will denote by
\begin{equation}\label{eq:jacminors}
J_{i,j}(f_1,f_2):=
\left|
\begin{array}{cc}
	\partial_i f_1 & \partial_j f_1 \\
	\partial_i f_2 & \partial_j f_2 	
\end{array}
\right|.	
\end{equation}
In the sequel, given a (homogeneous) polynomial $p(\ux)$, for all $j=1,\ldots,4$ we will denote by $\overline{p}^j$ the  polynomial $p$ in which the variable $x_j$ is set to zero. 
\begin{cor}\label{cor:def} Suppose given a commutative ring $R$, two positive integers $d_1,d_2$ such that $d_1+d_2\geq 3$ and two homogeneous polynomials $g_1,g_2$ in $R[\ux]$ of degree $d_1,d_2$ respectively. Then, 
	\begin{multline}\label{eqn:def}
	\Res(g_1,g_2,J_{1,2}(g_1,g_2),J_{2,3}(g_1,g_2))= 	(-1)^{d_1d_2} 	 \\ 
\cdot \Disc(g_1,g_2)  \Res(g_1,\partial_2g_1,g_2,\partial_2g_2)\Disc\left(\overline{g_1}^4,\overline{g_2}^4\right).
	\end{multline}
\end{cor}
\begin{proof} Straightforward by applying the formula in Theorem \ref{thm:gendef} with $l=x_4$, $m=x_3$, and $n=x_1$. We notice that
	\begin{equation}\label{eq:reductionvar}
	\Disc(f_1,f_2,x_4)=(-1)^{d_1d_2}\Disc\left(\overline{f}_1^4,\overline{f}_2^4\right)		
	\end{equation}
by property of the discriminant of three homogeneous polynomials in four variables \cite[Proposition 3.13]{BJ14}.
\end{proof}

From a computational point of view the above formula allows to compute the discriminant of any couple of homogeneous polynomials $g_1,g_2 \in R[\ux]$ as a ratio of determinants since all the other terms in \eqref{eq:reductionvar} can be expressed as ratio of determinants by means of the Macaulay formula. There is no need to introduce new variables as in the Cayley trick and the formula is universal in the coefficients of the polynomials over the integers.

\section{Properties and computational rules}\label{sec:formulas}

In this section, we provide some properties and computational rules of the discriminant $\Disc(f_1,f_2)$ as defined in the previous section. In particular, we give precise formulas regarding the covariance and invariance properties. We also provide a detailed computation of a particular class of complete intersection curves in order to illustrate how our formalism allows to handle the discriminant and simplify its computation and evaluation over any ring of coefficients. In what follows, $R$ denotes a commutative ring.

\subsection{First elementary properties}

From Theorem \ref{thm:gendef}, it is clear that the discriminant $\Disc(f_1,f_2)$ is homogeneous with respect to the coefficients of $f_1$, respectively $f_2$ and that these degrees can easily be computed. As expected, we recover the degrees of the usual geometric definition of discriminant (see \cite{Sylvester,SylvesterBis,Benoist}).

\begin{prop}[Homogeneity]\label{prop:homdeg} The universal discriminant is homogeneous of degree $\delta_i$ with respect to the coefficient of $f_i$ where, setting  $e_1=d_1-1$ and $e_2=d_2-1$, 
	$$\delta_1=d_2(3e_1^2+2e_1e_2+e_2^2), \ \  \delta_2=d_1(3e_2^2+2e_1e_2+e_1^2).$$
\end{prop}
\begin{proof} This is a straightforward computation from the defining equality (see Theorem \ref{thm:gendef}), since the degrees of resultants and discriminants of finitely many points are known (see \cite[Proposition 2.3]{J91} and \cite[Proposition 3.9]{BJ14}).
\end{proof}

\begin{prop}[Permutation of the polynomials]
Let $g_1,g_2 \in R[x_1,\ldots,x_4]$ be two homogeneous polynomials of degree $d_1$ and $d_2$ respectively, then
$$  \Disc\left({g_2},{g_1}\right)   =   \Disc\left({g_1}, {g_2}\right).$$	
\end{prop}
\begin{proof} This is a straightforward consequence of the similar property for resultants \cite[\S 5.8]{J91} and discriminants of finitely many points 	\cite[Proposition 3.12 i)]{BJ14}.
\end{proof}

\begin{prop}[Elementary transformations]\label{prop:elemtransform}
Let $g_1,g_2,h_1,h_2$ be four homogeneous polynomials in $R[\ux]$ of degree $d_1,d_2,d_1-d_2,d_2-d_1$ respectively. Then, 
$$\Disc\left(g_1,g_2+h_2g_1\right)   = \Disc\left(g_1+h_1g_2,g_2\right)   =   \Disc\left(g_1,g_2\right).$$ 	
\end{prop}

\begin{proof}
This is a straightforward consequence of the invariance of resultants under elementary transformations \cite[\S 5.9]{J91} and the invariance of discriminants of finitely many points under elementary transformations \cite[Proposition 3.12]{BJ14}.	
\end{proof}

\subsection{Covariance and invariance} 

In this section, we give precise statements about two important properties of the discriminant: its geometric covariance and its geometric invariance under linear change of  variables. 

\begin{prop}[Covariance]\label{prop:covariance}
Suppose given two homogeneous polynomials $g_1,g_2$ in $R[\ux]$ of the same degree $d\geq 2$ and a square matrix $\varphi=(u_{i,j})_{i,j=1,2}$
with coefficients in $R$, then 
\begin{multline*}
\Disc(u_{1,1}g_1+u_{1,2}g_2,u_{2,1}g_1+u_{2,2}g_2)=\\
\det(\varphi)^{6d(d-1)^2}\Disc(g_1,g_2).
\end{multline*}	
\end{prop}

\begin{proof} By definition, it is sufficient to prove this formula in the universal setting. For simplicity, we use the formula \eqref{eqn:def}.  
Setting $\tilde{f}_1:=u_{1,1}f_1+u_{1,2}f_2$ and $\tilde{f}_2=u_{2,1}f_1+u_{2,2}f_2$, we observe that
$J_{k,l}(\tilde{f}_1,\tilde{f}_2)=\det(\varphi)J_{k,l}(f_1,f_2)$
so that 
\begin{multline*}
\Res(\tilde{f}_1,\tilde{f}_2,J_{1,2}(\tilde{f}_1,\tilde{f}_2),J_{2,3}(\tilde{f}_1,\tilde{f}_2)) \\ =	\det(\varphi)^{4d^2(d-1)}\Res(\tilde{f}_1,\tilde{f}_2,J_{1,2},J_{2,3}).
\end{multline*}
In addition, by the covariance of resultants \cite[\S 5.11]{J91}, 
$$\Res(\tilde{f}_1,\tilde{f}_2,J_{1,2},J_{2,3})=\det(\varphi)^{4d(d-1)^2}\Res(f_1,f_2,J_{1,2},J_{2,3})$$
so that we deduce that
\begin{multline*}
\Res(\tilde{f}_1,\tilde{f}_2,J_{1,2}(\tilde{f}_1,\tilde{f}_2),J_{2,3}(\tilde{f}_1,\tilde{f}_2)) \\ =	\det(\varphi)^{4d(d-1)(2d-1)}\Res(f_1,f_2,J_{1,2},J_{2,3}).
\end{multline*}
The covariance of resultants also shows that
\begin{multline*}
\Res(\tilde{f}_1,\partial_2\tilde{f_1},\tilde{f}_2,\partial_2\tilde{f_2})\\
=\det(\varphi)^{d(d-1)(2d-1)}\Res(f_1,\partial_2f_1,f_2,\partial_2f_2)
\end{multline*}
and the covariance property of discriminants of finitely many points \cite[Proposition 3.18]{BJ14} yields
$$\Disc\left(\overline{\tilde{f}_1}^4,\overline{\tilde{f}_2}^4\right)=\det(\varphi)^{3d(d-1)}\Disc(\overline{f}_1^4,\overline{f}_2^4).$$
From all these equalities and \eqref{eqn:def}, we deduce the claimed formula.
\end{proof}

\begin{prop}[Invariance]\label{prop:changecoord}
Let $g_1,g_2$ be two homogeneous polynomials in $R[\ux]$ of degree $d_i \geq 2$ and let  
$\varphi = \left(c_{i,j}\right)_{1\leqslant i,j  \leqslant 4 }$
be a square matrix with entries in $R$.
For all homogeneous polynomial $g\in R[\ux]$ we set 
$$ g\circ \varphi (x_1,x_2,x_3,x_4):= g\left( \sum_{j=1}^4c_{1,j}x_j, \ldots, \sum_{j=1}^4c_{4,j}x_j \right).$$
Then, we have that
$$\Disc(g_1\circ \varphi,g_2\circ \varphi)
= \det(\varphi)^{D}\Disc(g_1,g_2)		
$$	
where $D=d_1d_2\left((e_1+e_2)^2-e_1e_2\right)$, $e_i=d_i-1$, $i=1,2$.
\end{prop}

	\begin{proof} As always, to prove this formula we may assume that we are in the universal setting, $f_1$ and $f_2$ being the universal homogeneous polynomials of degree $d_1$ and $d_2$ respectively. We will also denote by $l,m,n$ three generic linear form and by $\phi$ the generic square matrix of size $4$.
		
		Applying Theorem \ref{thm:gendef}, we get the equality 
\begin{multline}\label{eq:maincompose}
\Res(f_1\circ \varphi,f_2\circ \varphi,J(f_1\circ \varphi,f_2\circ \varphi, l\circ \varphi, m\circ \varphi),\\
\hspace{4.3cm} J(f_1\circ \varphi,f_2\circ \varphi, l\circ \varphi, n\circ \varphi))\\
= \Disc(f_1\circ \varphi,f_2\circ \varphi) \Disc\left(f_1\circ \varphi,f_2\circ\varphi, l\circ \varphi \right) \Res(f_1\circ \varphi,\\
\hspace{.7cm}J(f_1\circ \varphi, l\circ \varphi, m\circ \varphi, n\circ \varphi), f_2 \circ \varphi,J(f_2\circ \varphi,l\circ \varphi,m\circ \varphi,n\circ \varphi))	
\end{multline}
(observe that $l\circ\varphi, m\circ\varphi, n\circ\varphi$ are all linear forms in $\ux$).
Now, by \cite[Proposition 3.27]{BJ14}, we know that 
$$
\Disc\left(f_1\circ \varphi,f_2\circ\varphi, l\circ \varphi \right)
= \det(\varphi)^{d_1d_2(e_1+e_2)}\Disc(f_1,f_2,l).	
$$
 Also, by the chain rule formula for the derivative of the composition of functions, we have the formulas 
\begin{align*}
J(f_1\circ \varphi, f_2\circ \varphi, l\circ \varphi, m\circ \varphi))&= J(f_1,f_2,l,m) \circ [\varphi]\cdot \det(\varphi)	\\
J(f_1\circ \varphi, f_2\circ \varphi, l\circ \varphi, n\circ \varphi))&= J(f_1,f_2,l,n) \circ [\varphi]\cdot \det(\varphi)	\\
J(f_i\circ \varphi, l\circ \varphi, m\circ \varphi, n\circ \varphi))&= J(f_i,l,m,n) \circ [\varphi]\cdot \det(\varphi)	
\end{align*}
from we deduce, using the invariance of resultants \cite[\S 5.13]{J91} and their homogeneity, that
\begin{multline*}
 \Res(f_1\circ \varphi,J(f_1\circ \varphi, l\circ \varphi, m\circ \varphi, n\circ \varphi),f_2 \circ \varphi, \\
 \hspace{3.5cm}J(f_2\circ \varphi,l\circ \varphi,m\circ \varphi,n\circ \varphi))\\
 =\det(\varphi)^{d_1d_2e_1e_2} \Res(f_1,J(f_1,l,m,n),f_2,J(f_2,l,m,n))
\end{multline*}
and
\begin{multline*}
 \Res(f_1\circ \varphi,f_2\circ \varphi,J(f_1\circ \varphi,f_2\circ \varphi, l\circ \varphi, m\circ \varphi),\\
 \hspace{4cm} J(f_1\circ \varphi,f_2\circ \varphi, l\circ \varphi, n\circ \varphi))\\
 =\det(\varphi)^{d_1d_2(e_1+e_2)^2}  \Res(f_1,f_2, J(f_1,f_2,l,m), J(f_1,f_2,l,n)).
\end{multline*}
From here, the claimed formula follows from the substitution of the above equalities in \eqref{eq:maincompose} and the comparison with the formula given in Theorem \ref{thm:gendef}. 
 \end{proof}
 
 \begin{cor} The discriminant is invariant under permutation of the variables $\ux$.
 \end{cor}
 \begin{proof} It follows from Proposition \ref{prop:changecoord} since $D$ is even. 
 \end{proof}

\subsection{Discriminant of a plane curve}

Given a plane curve, we prove that its discriminant as defined in Section \ref{sec:def}, is compatible with its discriminant as a plane hypersurface \cite[\S 4.2]{BJ14}. 

\begin{lem}\label{lem:redvar}
 Let $g$ be a homogeneous polynomial in $R[\ux]$ of degree $d\geq 2$. Then, for all $i=1,\ldots,4$ we have that
$$\Disc(g,x_i)=\Disc\left(\overline{g}^i\right).$$	
\end{lem}
\begin{proof} By definition, it is sufficient to prove this equality in the case where $g$ is replaced by the generic homogeneous polynomial $f$ of degree $d$. We apply Theorem \ref{thm:gendef} with $l=x_r$, $m=x_s$, $n=x_t$ that are chosen so that $\{x_i,x_r,x_s,x_t\}=\{x_1,x_2,x_3,x_4\}$ as sets. We obtain the equality
\begin{equation*}
\Rc:=\Res(f,x_i,\pm \partial_t f, \pm \partial_s f)
=\Disc(f,x_i)\Disc\left(f,x_i,x_r\right).
\end{equation*}
Since the degree of $f$ and one of its partial derivative are consecutive integers, their product is always an even integer. It follows by 
standard properties of resultants that $\Rc$ does not depend on the sign of its entry polynomials, nor on their order, nor on the reduction of the variables, so that we have
$$
\Rc
=\Res(\overline{f}^i,\partial_t \overline{f}^i,\partial_s \overline{f}^i)=\Res(\overline{f}^i,\partial_s \overline{f}^i,\partial_t \overline{f}^i).
$$
Now, by property of discriminants, in particular \eqref{eq:reductionvar} and its invariance under permutation of variables \cite[Proposition 3.12]{BJ14}, we have
$$\Disc\left(f,x_i,x_r\right)=(-1)^{d}\Disc\left(\overline{f}^i,x_r\right)=\Disc\left(\overline{f}^{i,r}\right).$$
Finally, \cite[Proposition 4.7]{BJ14} shows that 
$$\Disc\left(\overline{f}^{i,r}\right) \Disc(\overline{f}^i)= \Res\left(\overline{f}^i, \partial_r \overline{f}^i, \partial_s \overline{f}^i\right)$$
and the claimed equality is proved.
\end{proof}

 \begin{prop}\label{prop:planecurve} Let $g \in R[\ux]$ be a homogeneous polynomial of degree $d \geq 2$ and $l=\sum_{i=1}^4 l_ix_i$ be a linear form in $R[\ux]$. Then, for all $i=1,\ldots,4$ we have that
	 \begin{multline*}
	 l_i^{2d(d-1)^2}\Disc(g,l) \\
	 =\Disc\left(g(l_ix_1-\delta_{i}^1l(\ux),\cdots,l_ix_4-\delta_{i}^4l(\ux))\right)
	 \end{multline*} 	
	where $\delta_i^j$ stands for the Kronecker symbol. 
 \end{prop}
 
\begin{proof} We assume that we are in the generic setting, which is sufficient to prove this corollary. Consider the linear change of coordinates given by the matrix $\varphi_i$ defined as follows: its $i^\mathrm{th}$ row is the vector $(-l_1 \ -l_2 \ -l_3 \ -l_4)^T$ and its other rows are filled with zeros except on the diagonal where we put $-l_i$. Then, 
it is not hard to check that 
$$g\circ\varphi_i\circ\varphi_i(\ux)=g(l_i^2\cdot\ux)=l_i^{2d}g(\ux).$$
Therefore, by Proposition \ref{prop:homdeg} we obtain
\begin{equation}\label{eq:discplancurve}
\Disc(g\circ\varphi_i\circ\varphi_i,l)= \Disc(l_i^{2d}g,l)
=l_i^{6d(d-1)^2}\Disc(g,l).
\end{equation}
On the other hand, since $l=x_i\circ\varphi(\ux)$, Proposition \ref{prop:changecoord} yields
\begin{multline*}
	\Disc(g\circ\varphi_i\circ\varphi_i,l)=\det(\varphi_i)^{d(d-1)^2}\Disc(g\circ\varphi_i,x_i)\\
	=l_i^{4d(d-1)^2}\Disc(g\circ\varphi_i,x_i)
\end{multline*}
(notice that $d(d-1)$ is even and $\det(\varphi_i)=-l_i^4$). Then, using Lemma \ref{lem:redvar} we deduce that
$$\Disc(g\circ\varphi_i\circ\varphi_i,l)  =  l_i^{4d(d-1)^2}\Disc(\overline{g\circ\varphi_i}^i).$$
Compared with \eqref{eq:discplancurve}, this latter equality shows that
$$l_i^{2d(d-1)^2}\Disc(g,l)=\Disc(\overline{g\circ\varphi_i}^i)$$
since $l_i$ is not a zero divisor in the universal ring of coefficients. Finally, to conclude we observe that
\begin{eqnarray*}
\lefteqn{\Disc(\overline{g\circ\varphi_i}^i)}\\
&=&	\Disc\left(g(-l_ix_1,\cdots,l(\ux)-l_ix_i,\cdots,-l_ix_4)\right) \\
&=& (-1)^{3d(d-1)^2}\Disc\left(g(l_ix_1,\cdots,l_ix_i-l(\ux),\cdots,l_ix_4)\right)
\end{eqnarray*}
where the last equality follows from the homogeneity of the discriminant of a single  polynomial \cite[Proposition 4.7]{BJ14}.
\end{proof}

\subsection{A sample calculation}

In order to illustrate the gain we obtain with the new formalism we are developing, we give an explicit decomposition of the discriminant of a particular family of complete intersection space curves that are drawn on a generalized cylinder whose base is an arbitrary algebraic plane curve.

\begin{prop}\label{prop:samplecomp}
Suppose given an element $u\in R$ and two homogeneous polynomials $f, g \in R[x_1,x_2,x_3]$ of degree $d_1$ and $d_2$ respectively. If $d_1+d_2\geq 3$ then 
\begin{multline*}
	\Disc(ux_4^{d_1}+f(x_1,x_2,x_3),g(x_1,x_2,x_3)) = \\ 
(-1)^{d_1}d_1^{d_1d_2(d_1+d_2-3)} u^{d_2[(d_1+d_2-2)^2-(d_1-1)(d_2-1)]}\\ 
\cdot \Disc(f,g)^{d_1-1}\Disc(g)^{d_1}.
\end{multline*}	
\end{prop}

\begin{proof} Because of the space limitation, we will only give the main lines to prove this formula. First, we notice that it is sufficient to assume that we are in the universal setting, that is to say to assume that the coefficients of $f,g$ and $u$ are indeterminates over the integers.
	
	Set $f_1=ux_4^{d_1}+f$ and $f_2=g$. By Corollary \ref{cor:def}, we have that 
	\begin{multline}\label{eq:ufromula1}  
		\Res(f_1,f_2, J_{1,2}, J_{2,3}) =(-1)^{d_1d_2}\Disc \left(f_1,f_2 \right) \\
		 \cdot \Res(f_1,\partial_2f_1,f_2, \partial_2f_2) \Disc(f,g).  
	\end{multline}
	Applying Laplace's formula \cite[\S 5.10]{J91}, we get 
$$\Res(f_1,\partial_2f_1, f_2, \partial_2 f_2)=u^{d_2(d_1-1)(d_2-1)} \Res(g,\partial_2f, \partial_2g)^{d_1}$$
	and substituting this equality in \eqref{eq:ufromula1}, we deduce that
	\begin{multline} \label{eq:uform5}
		\Res(f_1,f_2, J_{1,2}, J_{2,3}) = (-1)^{d_1d_2}u^{d_2(d_1-1)(d_2-1)} \\
		 \cdot \Disc \left(f_1,f_2 \right) \Res(g,\partial_2f, \partial_2g)^{d_1} \Disc(f,g).  \end{multline}
	Now, applying again Laplace's formula we get that	
	\begin{eqnarray}\label{eq:uform2}
	\lefteqn{\Res(f_1,f_2, J_{1,2}, J_{2,3})} \\ \nonumber
	&=&   u^{d_2(d_1+d_2-2)^2} \Res(g, J_{1,2} (f,g), J_{2,3}(f,g))^{d_1}.
	\end{eqnarray}
	In order to compute $\mathcal{R}:= \Res(g, J_{1,2}(f,g),J_{2,3}(f,g))$, we first observe that
	$$ \Res(g,x_1J_{1,2}(f,g), J_{2,3}(f,g))= \Res(g,x_1,J_{2,3}(f,g)) \mathcal{R}$$
	by multiplicativity of resultants. From the definition of the Jacobian determinants we have 
	$$x_1J_{1,2}(f,g)=d_1f\partial_2g-d_2g\partial_2f +x_3J_{2,3}(f,g)$$
	and we deduce that 
	\begin{eqnarray*}
	\lefteqn{  \Res(g,x_1J_{1,2}(f,g), J_{2,3}(f,g)) = \Res(g, d_1f \partial_2g, J_{2,3}(f,g))} \\
	&=& d_1^{d_2(d_1+d_2-2)}\Res( g,f, J_{2,3}(f,g)) \Res(g,\partial_2g, \partial_2f) \\
	& & \hspace{5cm} \cdot \Res(g,\partial_2g,\partial_3g). 
	\end{eqnarray*}
	But from the rule of permutation of polynomials for resultants \cite[\S 5.8]{J91} and the definition of discriminants of finitely many points \cite[Definition 3.5]{BJ14}, we have
	\begin{align*}
	\Res( g,f, J_{2,3}(f,g)) &=(-1)^{d_1+d_2}\Res(f,g,J_{2,3}(f,g))\\
	&=\Disc(f,g)\Res(\overline{f}^1,\overline{g}^1).	
	\end{align*}
	Similarly, from the rule of permutations of polynomials and the definition of discriminants of hypersurfaces \cite[Definition 4.6]{BJ14}, we have
	$\Res(g,\partial_2g,\partial_3g)=\Disc(g)\Disc(\overline{g}^1)$ and hence
	\begin{multline}\label{eq:factorR}
	\Res(g,x_1,J_{2,3}(f,g)) \mathcal{R} = (-1)^{d_1+d_2}d_1^{d_2(d_1+d_2-2)} 
	\\ \cdot \Res(g,\partial_2g, \partial_2f)
	 \Disc(f,g)\Disc(g)  \Res(\overline{f}^1,\overline{g}^1)\Disc(\overline{g}^1).
	\end{multline}
	Now, it remains to compute $\mathcal{R}'=\Res(g,x_1,J_{2,3}(f,g)).$ On the one hand we have
	\begin{equation}\label{eq:uform3}
	\Res(g,x_1,x_2J_{2,3}(f,g))= \Res(g,x_1,x_2) \mathcal{R}'.	
	\end{equation}
	On the other hand,
	\begin{multline}\label{eq:uform3b}
	\Res\left(g,x_1,x_2J_{2,3}(f,g)\right) \\
	= (-1)^{d_1d_2}\Res\left(\overline{g}^1, x_2J_{2,3}\left(\overline{f}^1, \overline{g}^1\right)\right)	
	\end{multline}
	and since $x_2J_{2,3}(\overline{f}^1,\overline{g}^1)=d_1\overline{f}^1 \partial_3\overline{g}^1 - d_2 \overline{g}^1 \partial_3\overline{f}^1$
	by the Euler formula, we deduce that
	\begin{eqnarray}\label{eq:uform4}
	\lefteqn{\Res(\overline{g}^1, x_2J_{2,3}(\overline{f}^1, \overline{g}^1))} \\ \nonumber
	&=& (-1)^{d_2}d_1^{d_2} \Res(\overline{g}^1, \overline{f}^1) \Disc(\overline{g}^1) \Res(\overline{g}^1,X_2). 
	\end{eqnarray}
	Finally, since
	$\Res(\overline{g}^1,x_2)=(-1)^{d_2}\Res(g,x_1,x_2)$ 
	the comparison of \eqref{eq:uform3}, \eqref{eq:uform3b} and \eqref{eq:uform4} shows that
	\begin{align*}
	\mathcal{R}'
		&= d_1 ^{d_2} \Disc(\overline{g}^1)  \Res(\overline{f}^1, \overline{g}^1).
	\end{align*}
	Now, coming back to the factor $\mathcal{R}$, we deduce from \eqref{eq:factorR} that
	\begin{equation*} 
		\mathcal{R}=  (-1)^{d_1+d_2} d_1^{d_2(d_1+d_2-3)} \Res(g, \partial_2 g, \partial_2 f)\Disc(f,g) \Disc(g)     
	\end{equation*}
	and hence from \eqref{eq:uform2} that
	\begin{multline*} 
	\Res(f_1,f_2, J_{1,2}, J_{2,3})	=  (-1)^{d_1(d_2-1)} d_1^{d_1d_2(d_1+d_2-3)} 
	\\ \cdot u^{d_2(d_1+d_2-2)^2}  \Res(g, \partial_2 g, \partial_2 f)^{d_1}  \Disc(f,g)^{d_1} \Disc(g)^{d_1}.     
	\end{multline*}
	Finally, we deduce from (3) that
	\begin{eqnarray*}  
	\lefteqn{(-1)^{d_1d_2}  u^{d_2(d_1-1)(d_2-1)} \Disc \left(f_1,f_2 \right)   } \\
	& & \hspace{1.5cm} \cdot \Res(g,\partial_2f, \partial_2g)^{d_1} \Disc(f,g) \\
	&=&	 (-1)^{d_1(d_2-1)} d_1^{d_1d_2(d_1+d_2-3)} u^{d_2(d_1+d_2-2)^2}  \\
	& & \hspace{1.5cm} \cdot \Res(g, \partial_2 g, \partial_2 f)^{d_1} \Disc(f,g)^{d_1} \Disc(g)^{d_1} 
	\end{eqnarray*}
and the claimed formula follows.
\end{proof}

\section{The geometric property}\label{sec:geom}

The aim of this section is to show that the discriminant $\Disc(f_1,f_2)$ defined in Definition \ref{def:disc} satisfies to the expected geometric property, namely that its vanishing corresponds to the existence of a singular point on the curve intersection of the two surfaces of equations $f_1=0$ and $f_2=0$ in $\PP^3$. We start by recalling the precise meaning of this geometric property as we will work over coefficient rings which are not necessarily fields. 

\smallskip

Let $k$ be a commutative ring. We consider the universal setting over $k$, i.e.~we suppose given two positive integers $d_1,d_2$ and we consider the (generic) homogeneous polynomials in the four variables $\ux=(x_1,x_2,x_3,x_4)$ 
$$f_1:=\sum_{|\alpha|=d_1} U_{1,\alpha}x^\alpha, \ \ f_2:=\sum_{|\alpha|=d_2} U_{2,\alpha}x^\alpha$$
that are polynomials in ${}_k C={}_kA[\ux]$, where ${}_k A=k[U_{i,\alpha}; i=1,2, |\alpha|=d_i]$ is the universal ring of coefficients over the base ring $k$. If there is no possible confusion, we will omit the subscript $k$ in the above notation.

We define the ideal $\mm=(\ux)\subset C$ generated by the variables $\ux$, the ideal $\Jc \subset C$ generated by all the 2-minors of the Jacobian matrix of $f_1$ and $f_2$, and the ideal 
$\Dc=(f_1,f_2)+\Jc \subset C$. Thus, using notation \eqref{eq:jacminors}, we have that
$$\Dc=(f_1,f_2,J_{1,2},J_{1,3},J_{1,4},J_{2,3},J_{2,4},J_{3,4})\subset C.$$
The quotient ring $B:= C/ \Dc$ is a graded ring with respect to the variables $\ux$. As such, it gives rise to the projective scheme $\Proj(B)\subset \PP^3_{A}$ that corresponds to the points $ (u_{i,\alpha})_{i,\alpha}  \times P \in \Spec(A)\times \PP^3_k$ such that the corresponding polynomials $f_1,f_2$ and all the 2-minors of their Jacobian matrix vanish simultaneously at $P$. The canonical projection of $\Proj(B)$ onto $\Spec(A)$ is a closed subscheme $\Delta$ of $\Spec(A)$ whose support is precisely what is commonly called the \emph{discriminant locus}. By definition, the defining ideal of $\Delta$ is the ideal $\P=\T_\mm(\Dc)\cap A$ where 
\begin{align*}
	\T_\mm(\Dc) &= \ker\left( C\xrightarrow{\pi} \prod_{i=1}^4 B_{x_i} \right) \subset C \\
	&= \left\{ p \in C \textrm{ such that } \exists \nu \in \NN : \mm^\nu \cdot p \subset \Dc\right\}
\end{align*}
is the so-called ideal of \emph{inertia forms} -- the notation $B_{x_i}$ denotes the localization of $B$ with respect to the variable $x_i$ and $\pi$ is the product of the canonical quotient maps.

\smallskip

In what follows, we will show that ${}_k \Disc(f_1,f_2)$, as defined by Definition \ref{def:disc}, is a generator of ${}_k \P$ if $k$ is a UFD, so that it satisfies to the expected geometric property. Before going into the details, we recall the following important and well-known result (see e.g.~\cite{Benoist,GKZ}): if $k$ is a field, then the reduced scheme of ${}_k \Delta$ is an irreducible hypersurface, i.e.~the radical of ${}_k \P$ is a principal and prime ideal, so that it is generated by an irreducible polynomial $D_k \in {}_k A$. This polynomial is not unique; it is unique up to multiplication by a nonzero element in $k$. In addition, $D_k$ is homogeneous of degree $\delta_i$ (see Proposition \ref{prop:homdeg} for the definition of $\delta_i$) with respect to the coefficients of $f_i$. 

\smallskip

We begin with some preliminary results on the Jacobian minors and the ideal $\Jc$ they generate.
\begin{lem}\label{lem:Jdet}	
	For any $j\in \{1,\ldots,4\}$ we have that
	$$\sum_{k\in \{1,\ldots,4\}, k\neq j} x_kJ_{k,j} \in (d_1f_1,d_2f_2).$$
\end{lem}
\begin{proof} Using the Euler formula, we have that 
\begin{align*}	
\left|
\begin{array}{cc}
	d_1f_1 & \partial_jf_1 \\
	d_2f_2 & \partial_jf_2 \\
\end{array}
\right|
& =
\left|
\begin{array}{cc}
	\sum_{i=1}^4x_i\partial_if_1 & \partial_jf_1 \\
	\sum_{i=1}^4x_i\partial_if_1 & \partial_jf_2 \\
\end{array}
\right| \\ 
&=\sum_{k\in \{1,\ldots,4\}, k\neq j} x_k 
\left|
\begin{array}{cc}
	\partial_k f_1 & \partial_jf_1 \\
	\partial_k f_2 & \partial_jf_2 \\
\end{array}
\right| 	
\end{align*}
and the claim follows.
\end{proof}

\begin{lem}\label{lem:ENrel} For any integer $j\in\{1,2\}$ and any triple of distinct integers $i_1,i_2,i_3$ in $\{1,2,3,4\}$ we have that 
	$$J_{i_2,i_3}.\partial_{i_1}f_j-J_{i_1,i_3}.\partial_{i_2}f_j+J_{i_1,i_2}.\partial_{i_3}f_j=0.$$
\end{lem}
\begin{proof} Develop the determinant of the 3-minor corresponding to the columns $i_1,i_2,i_3$ in the Jacobian matrix of $f_1,f_2$ and $f_j$. 
\end{proof}

\begin{lem}\label{lem:primenessJ} If $k$ is a domain, then for all $i=1,\ldots,4$ the ideal $\Jc_{x_i} \subset {}_k C_{x_i}$ is a prime ideal.
\end{lem}
\begin{proof} For simplicity, we will assume that $i=4$, the other cases being similar. In order to emphasize some particular coefficients of $f_1$ and $f_2$ we rewrite them as follows:
	$$f_i=U_{i,0}x_4^{d_i}+x_4^{d_i-1}(U_{i,1}x_1+U_{i,2}x_2+U_{i,3}x_3)+h_i, \, i=1,2.$$
We consider the $A$-algebra morphism
$$\eta:C[x_4^{-1}] \rightarrow C[x_4^{-1}] : U_{i,j} \mapsto -\partial_jf_i/x_4^{d_i-1}$$
which leaves invariant all the variables $\ux$ and all the coefficients of $f_1,f_2$, except the $U_{i,j}$'s. 
As $\eta(\partial_j f_i)=-U_{i,j}x_4^{d_i-1}$, $\eta$ is surjective. Moreover, setting 
$$\UU=\left( 
\begin{array}{ccc}
	U_{1,1} & U_{1,2} & U_{1,3} \\
	U_{2,1} & U_{2,2} & U_{2,3}
\end{array}
\right)$$
and denoting by $\UU_{r,s}$ the 2-minor of $\UU$ corresponding to the column number $r,s$, we have that
$\eta(J_{r,s})=x_4^{d_1+d_2-2}\UU_{r,s},$ $r,s \in \{1,2,3\}, r\neq s.$
Considering the map $\overline{\eta}$ induced by $\eta$, 
$$\overline{\eta} :  C[x_4^{-1}] \rightarrow C[x_4^{-1}]/\left(\UU_{1,2},\UU_{2,3},\UU_{1,3}\right),$$
we deduce that $\left(J_{1,2},J_{2,3},J_{1,3}\right).C[x_4^{-1}]\subset \ker(\overline{\eta}).$
Actually, this inclusion is an equality. Indeed, if $p \in \ker(\overline{\eta})$ then 
\begin{equation}\label{eq:inclusion}
\eta\left(p(U_{i,j})\right)=p\left(-\partial_jf_i/x_4^{d_i-1}\right) \in \left( \UU_{1,2},\UU_{2,3},\UU_{1,3} \right).
\end{equation}
But since $\eta(-\partial_jf_i/x_4^{d_i-1})=U_{i,j}$, applying again $\eta$ to \eqref{eq:inclusion} we deduce that
$$p(U_{i,j}) \in \eta(\UU_{1,2},\UU_{2,3},\UU_{2,3})=\left(J_{1,2},J_{2,3},J_{1,3}\right).C[x_4^{-1}].$$
It follows that $\overline{\eta}$ induces a graded isomorphism
\begin{equation}\label{eq:isoJ}
C_{x_4}/\Jc_{x_4}
\xrightarrow{\sim} C[x_4^{-1}]/\left(\UU_{1,2},\UU_{2,3},\UU_{1,3}\right).
\end{equation}
From here, if $k$ is a domain then the ideal generated by the 2-minors of $\UU$ is a prime ideal 
(see \cite[Theorem 2.10]{BV88}) and hence 
$C_{x_4}/\Jc_{x_4}$ is a domain. 
\end{proof}

The above lemma is the key result to deduce the following properties of the ideal of inertia forms $\T_\mm(\Dc)$.

\begin{prop}\label{prop:Bxi-domain}
	If $k$ is a domain then  $B_{x_i}$ is a domain for all $i=1,\ldots,4$.
\end{prop}
\begin{proof} We prove the claim for $i=4$, the other cases being similar. Let $p_1, p_2$ be two polynomials in $C$ so that $p_1p_2=0$ in $B_{x_4}$, i.e.~$p_1p_2$ belongs to the ideal $\Dc$ up to multiplication by a power of $x_4$. Using this fact and Lemma \ref{lem:Jdet}, we deduce that there exists an integer $\nu$ such that
\begin{equation}\label{eq:choicex4}
	x_4^\nu p_1p_2 \in (f_1,f_2,J_{1,2},J_{2,3},J_{1,3}).	
\end{equation}	
In order to emphasize the leading coefficients of $f_1$ and $f_2$ with respect to the variable $x_4$, we rewrite them as
$$f_i=U_{i,0}x_4^{d_i}+q_i, \ i=1,2.$$
Denote by $\overline{C}$ the polynomial ring $C$ in which the variables (coefficients) $U_{1,0}, U_{2,0}$ are removed and consider the surjective graded morphism
$$\rho:C[x_4^{-1}] \rightarrow \overline{C}[x_4^{-1}] : U_{i,0} \mapsto -q_i/x_4^{d_i}$$
which leaves invariant all the variables $\ux$ and all the coefficients of $f_1,f_2$, except $U_{1,0}, U_{2,0}$. It induces an isomorphism
$$\overline{\rho} : C_{x_4}/(f_1,f_2)_{x_4} \xrightarrow{\sim} \overline{C}[x_4^{-1}].$$
Now, by \eqref{eq:choicex4} we deduce that $\rho(p_1)\rho(p_2)$ belongs to the ideal $(J_{1,2},J_{1,3},J_{2,3}).\overline{C}[x_4^{-1}].$
Therefore, using Lemma \ref{lem:primenessJ} we deduce that either $\rho(p_1)$ or $\rho(p_2)$ belongs to this ideal, say $\rho(p_1)$. This implies that 
there exists an integer $\mu$ such that
$$x_4^\mu p_1 \in (f_1,f_2,J_{1,2},J_{2,3},J_{1,3}) \subset \Dc.$$	
In turns, this implies precisely that $p_1=0$ in $B_{x_4}$, which concludes the proof.
\end{proof}

\begin{cor}\label{cor:tf} 
	For all $i=1,\ldots,4$ we have that
	\begin{align*}
		\T_\mm(\Dc)  &= \ker\left( C\xrightarrow{\pi_i} B_{x_i} \right) \\
		&= \{ p \in C \textrm{ such that } \exists \nu \in \NN : x_i^\nu\cdot p \in \Dc\}.
	\end{align*}
Thus, both $\T_\mm(\Dc)$ and $\P$ are prime ideals if $k$ is a domain.
\end{cor}
\begin{proof} Using Proposition \ref{prop:Bxi-domain}, the proof of \cite[Corollary 3.21]{BJ14} applies verbatim to show that $x_i$ is not a zero divisor in $B_{x_j}$ for all $i\neq j$. From here, we deduce that the canonical maps $B_{x_i}\rightarrow B_{x_ix_j}$, $i\neq j$, are all injectives maps and hence the claimed equalities follow. 
\end{proof}

We are now ready to prove the main result of this section.

\begin{thm} If $k$ is a UFD then ${}_k\Disc(f_1,f_2)$ is a generator of ${}_k \P$. It is hence an irreducible polynomial in ${}_k A$.
\end{thm}
\begin{proof} 
First, let $\KK$ be a field. From the geometric property we recalled previously, we know that the radical of ${}_\KK \P$ is generated by an irreducible  polynomial $D_\KK$. Using Corollary \ref{cor:tf}, we deduce that  $D_\KK$ is actually a generator of ${}_\KK \P$.	
	
Now, assume that $k$ is a domain and take again the notation of Theorem \ref{thm:gendef}. The resultant
$$\Res\left(f_1,f_2, J(f_1,f_2,l,m), J(f_1,f_2,l,n)\right)$$ 
is an inertia form of its four input polynomials and hence, by developing the Jacobian determinants, we see that it belongs to ${}_k\P\cdot {}_kA'$.
Therefore, Theorem \ref{thm:gendef} shows that
	\begin{multline}\label{eq:DinP1}
		\Res\left(f_1,J(f_1,l,m,n)),f_2,J(f_2,l,m,n)\right)\\
		\cdot \Disc(f_1,f_2)\Disc\left({f_1},{f_2},l\right) \in {}_k\P\cdot {}_kA'.
	\end{multline}
We claim that $\Disc(f_1,f_2,l)$ does not belong to ${}_k\P$. Indeed, assume the contrary. By extension to the fraction field $\KK$ of $k$, we would have that the square-free part of ${}_\KK\Disc(f_1,f_2,l)$ belongs to the prime ideal ${}_\KK \P$. But ${}_\KK \P$ is generated by $D_\KK$ so we get a contradiction since ${}_\KK\Disc(f_1,f_2,l)$ is homogeneous of degree $d_2(e_1+2e_2)<\delta_1$ with respect to the coefficients of $f_1$, and similarly with respect to the coefficients of $f_2$ \cite[Proposition 3.9]{BJ14}.  With a similar argument, we also get that the resultant 
$\Res\left(f_1,J(f_1,l,m,n),f_2,J(f_2,l,m,n)\right)$
does not belong to ${}_k \P$ since it is homogeneous of degree $d_2(2e_1e_2+e_2)$ with respect to the coefficients of $f_1$, and similarly with respect to the coefficients of $f_2$. Finally, as ${}_k \P$ is a prime ideal we deduce from \eqref{eq:DinP1} that ${}_k\Disc(f_1,f_2)$ belongs to ${}_k \P$ (recall that $k$ is here assumed to be a domain). 

As a first consequence, since ${}_k \Disc(f_1,f_2)$ and $D_k$ are homogeneous polynomials of the same degree with respect to the coefficients of each $f_i$, we conclude that this theorem is proved if $k$ is a assumed to be a field.

\smallskip	

Let $p\in{}_\ZZ \P$ be an inertia form and set $N:=d_1d_2(e_1+e_2)$. Our next aim is to show that	
${}_\ZZ\Disc(f_1,f_2)$ divides $p^{N}$. For that purpose, using the definition of inertia forms and Lemma \ref{lem:Jdet}, we deduce that there exists an integer $\nu$ such that
	$x_4^\nu \cdot p \in \left(f_1,f_2,J_{1,2},J_{2,3},J_{1,3}\right).$
Then, Lemma \ref{lem:ENrel} shows that $\partial_2f_i\cdot J_{1,3} \in \left(J_{1,2},J_{2,3}\right)$ so that we get that
	$$x_4^\nu\cdot p \cdot \partial_2f_i \in \left(f_1,f_2,J_{1,2},J_{2,3}\right).$$
Now, by the divisibility property of resultants \cite[\S 5.6]{J91} we obtain that $\Res(f_1,f_2,J_{1,2},J_{2,3})$ divides
$$\Rc:=\Res(f_1,f_2,J_{1,2},x_4^\nu\cdot p \cdot \partial_2f_i).$$
Applying computational rules of resultants and choosing $i=1$ we get that
\begin{multline*}
\Rc=p^{N}\Res\left(\overline{f_1}^4,\overline{f_2}^4,\overline{J_{1,2}}^4\right)^\nu\\
\cdot \Res(f_1,\partial_2f_1,f_2,\partial_2f_2)\Res(f_1,f_2,\partial_1f_1,\partial_2f_1)
\end{multline*}
where, in addition,
$$\Res\left(\overline{f_1}^4,\overline{f_2}^4,\overline{J_{1,2}}^4\right)=\Disc\left(\overline{f_1}^{4},\overline{f_2}^{4}\right)\Res\left(\overline{f_1}^{3,4},\overline{f_2}^{3,4}\right)$$
by the definition of discriminants of finitely many points \cite[Definition 3.5]{BJ14}. Combining the above equalities and using \eqref{eqn:def}, we deduce that $\Disc(f_1,f_2)$ divides the product
\begin{multline}\label{eq:divdiscx4}
p^{N}\Disc\left(\overline{f_1}^{4},\overline{f_2}^{4}\right)^{\nu-1}\Res\left(\overline{f_1}^{3,4},\overline{f_2}^{3,4}\right)^\nu\\
\cdot \Res(f_1,f_2,\partial_1f_1,\partial_2f_1).	
\end{multline}
With a similar degree inspection as above and after extension to $\QQ$, we deduce that ${}_\QQ \Disc(f_1,f_2)$, which is an irreducible polynomial, cannot divide the discriminant and the two resultants in \eqref{eq:divdiscx4}. Then, we claim that the discriminant and the two resultants in \eqref{eq:divdiscx4} are primitive polynomials. This is a known property for the first two ones. For the third one, namely $\Res(f_1,f_2,\partial_1f_1,\partial_2f_1)$, we argue by specialization: for instance,  
\begin{align*}
	\Res(f_1,Ux_4^{d_2},\partial_1f_1,\partial_2f_1) 
	&= U^{d_1e_1^2}\Res\left(\overline{f_1}^4,\overline{\partial_1f_1}^4,\overline{\partial_2f_1}^4\right)^{d_2}\\
	&=U^{d_1e_1^2} \Disc\left(\overline{f_1}^4\right)^{d_2}\Disc\left(\overline{f_1}^{3,4}\right)^{d_2}
\end{align*}
is a primitive polynomial since both discriminants on the right side are known to be primitive polynomials. Therefore, we conclude that ${}_\ZZ \Disc(f_1,f_2)$ 
divides $p^{N}$.

\smallskip

Finally, from what we proved we deduce that ${}_k \P$ and the ideal generated by ${}_k \Disc(f_1,f_2)$ have the same radicals. Since we assume that $k$ is a UFD, ${}_k \P$ is prime and we deduce that there exist an irreducible polynomial $P_k$, an invertible element $c \in k$ and a positive integer $r$ such that ${}_k\Disc(f_1,f_2)=c\cdot P_k^r$. By extension to $\KK$ we deduce immediately that $r=1$, which concludes the proof.
\end{proof}

The above theorem shows that ${}_\ZZ \Disc(f_1,f_2)$ is a primitive and irreducible polynomial in ${}_kA$. It also shows that the discriminant formula we gave provides an \emph{effective} smoothness criterion (as the criterion in \cite[p.~3]{Dem} that applies verbatim).

\section{Acknowledgments}

Part of this work was done while the second author was visiting IMSP at Benin, supported by the DAAD. Both authors warmly thank the ICTP for its hospitality and are very grateful to Alicia Dickenstein, Marie-Fran\c{c}oise Roy and Fernando Rodrigues Villegas for their continued support.

{\small
\def\cprime{$'$}

}

\end{document}